\documentclass[conference]{IEEEtran}
\IEEEoverridecommandlockouts

\usepackage[hidelinks]{hyperref}
\usepackage{cite}
\usepackage{amsmath,amssymb,amsfonts}
\usepackage{pifont}
\usepackage{algorithm}
\usepackage{algorithmicx}
\usepackage{algpseudocodex} 
\usepackage{graphicx}
\usepackage{textcomp}
\usepackage{xcolor}
\definecolor{okabe1}{HTML}{000000}
\definecolor{okabe2}{HTML}{E69F00}
\definecolor{okabe3}{HTML}{56B4E9}
\definecolor{okabe4}{HTML}{009E73}
\definecolor{okabe5}{HTML}{F0E442}
\definecolor{okabe6}{HTML}{0072B2}
\definecolor{okabe7}{HTML}{D55E00}
\definecolor{okabe8}{HTML}{CC79A7}

\usepackage[]{todonotes} 

\usepackage{amsthm}
\usepackage{cleveref}

\newtheorem{theorem}{Theorem}
\newtheorem{lemma}[theorem]{Lemma}
\newtheorem{corollary}[theorem]{Corollary}

\newtheorem{definition}[theorem]{Definition}
\renewcommand{\emph}[1]{\textit{\textbf{#1}}}

\def\BibTeX{{\rm B\kern-.05em{\sc i\kern-.025em b}\kern-.08em
    T\kern-.1667em\lower.7ex\hbox{E}\kern-.125emX}}

\usepackage{tikz}
\usepackage{textcomp}
\usepackage{hyperref}
\usepackage{lipsum}

\newcommand\copyrighttext{%
  \footnotesize \textcopyright 2026 IEEE. Personal use of this material is permitted.
  Permission from IEEE must be obtained for all other uses, in any current or future
  media, including reprinting/republishing this material for advertising or promotional
  purposes, creating new collective works, for resale or redistribution to servers or
  lists, or reuse of any copyrighted component of this work in other works.}
\newcommand\copyrightnotice{%
\begin{tikzpicture}[remember picture,overlay]
\node[anchor=south,yshift=10pt] at (current page.south) {\fbox{\parbox{\dimexpr\textwidth-\fboxsep-\fboxrule\relax}{\copyrighttext}}};
\end{tikzpicture}%
}

\begin{document}

\title{Raiders of the Lost Log: Synchronous Parallel In-Place Models and Algorithms
\thanks{Research supported in part by NSF grant CCF-2212129.}
}

\author{\IEEEauthorblockN{Michael T. Goodrich}
\IEEEauthorblockA{\textit{Department of Computer Science} \\
\textit{UC Irvine}\\
Irvine, CA \\
\texttt{goodrich@uci.edu}}
\and
\IEEEauthorblockN{Vinesh Sridhar}
\IEEEauthorblockA{\textit{Department of Computer Science} \\
\textit{UC Irvine}\\
Irvine, CA \\
\texttt{vineshs1@uci.edu}}
}


\maketitle
\copyrightnotice

\begin{abstract}

Embedded systems and Internet of Things (IoT) applications motivate
\emph{in-place} parallel algorithms, which avoid allocating additional shared memory past the input. 
Work by Gu, Obeya, and Shun [APOCS '21] defines a family of PIP (parallel in-place) models and parallel algorithms that eschew auxiliary memory at high processor counts while remaining \emph{in-situ} when run sequentially.
However, their models assume asynchronous processing and have no in-place guarantees for intermediate processor counts. 
We address this gap in the literature by proposing a
\emph{Synchronous PIP} family of models for in-place parallel and distributed computation. We demonstrate the effectiveness of our new model by 
giving efficient and synchronous parallel algorithms in this model that require no auxiliary shared memory and only constant private memory per processor. 
Importantly, we show how to leverage a new \emph{parallel-augmented sweep} technique to ensure that Synchronous PIP algorithms remain efficient and strictly in-place at all processor counts. 

\end{abstract}


\pagestyle{plain}
\section{Introduction}

Much of the focus in parallel algorithm design has been on finding parallel algorithms that minimize work, the total number of operations performed, and span, the time until last processor terminates. 
However, data processing tasks have increased in size by orders of magnitude over the past decade. 
Machine rental and energy costs grow linearly with a machine's memory capacity~\cite{gu2021parallel}. 
Furthermore, 
memory bandwidth/latency, cache misses, and page faults all present larger bottlenecks to performance and scalability if shared memory usage is high. 
Thus, memory usage in Big Data tasks becomes a concern equally important to runtime.
This issue is magnified when dealing with shared-nothing data processing frameworks such as MapReduce, Hadoop, and Spark since shared working set memory is not locally accessible but rather distributed over the network. 
Many classic parallel algorithms, 
despite having excellent performance in theory, 
require
$\Omega(n)$ auxiliary memory (see, e.g.,~\cite{blelloch2010low,blelloch2020parallelism,blelloch2020optimal,goodrich1987efficient,amato1994parallel,zhang1991optimal}).
As a result, they scale poorly in practice.

\subsection{Prior Work}
In response, researchers have developed \emph{Parallel In-Place} (PIP) algorithms~\cite{gu2021parallel,guan1991time,guan1992parallel,langston1993time,obeya2019theoretically,zheng1999efficient}, 
which restrict auxiliary space usage. 
Notably, work by Gu, Obeya, and Shun~\cite{gu2021parallel} synthesizes past efforts and introduces two models for parallel
in-place processing: Strong PIP and Relaxed PIP. 
A shared goal of these models is to strike a balance between parallelism and space 
usage by requiring PIP algorithms 
to scale down to space-efficient sequential algorithms.
Specifically,
in the \emph{Strong PIP} model an algorithm (1) allocates no shared memory aside from the input, (2) uses $O(\log n)$ private memory per processor, (3) has polylogarithmic span, and (4) uses at most $O(\log n)$ space when run sequentially. 
In the \emph{Relaxed PIP} model an algorithm (1) allocates $O(t)$ auxiliary shared memory, (2) uses $O(\log n)$ private memory per processor, and (3) has $O((n/t) \text{polylog}(n))$ span, where $t$ denotes a cap on auxiliary memory. 
Importantly,
these PIP models are specializations of the
\emph{software parallel} asynchronous binary-forking model 
of parallel computation, which is based on binary ``fork'' operations that
exist in some parallel programming languages. 

In this paper, we are interested in \emph{hardware parallel} models
for in-place computations, however, where we have a specified number,
$p$, of processors, where $p$ lies somewhere in between $1$ and the input size, $n$. 
Indeed, providing efficient algorithms that do not allocate additional memory for these intermediate processor counts is the main motivation for our new model, \emph{Synchronous Strict PIP}, which is a specialization of the well-known hardware-parallel PRAM model of computation.
We define in-place variants of PRAM models in this paper; hence,
let us briefly review the PRAM model, which is a family of 
hardware-based synchronous
models of parallel computation.
In the Exclusive-Read, Exclusive-Write (EREW) PRAM model, 
processors cannot read
or write to the same shared memory location concurrently. In the
Concurrent-Read, Exclusive-Write (CREW) PRAM model, 
processors can read from shared memory concurrently
but cannot write to the same location concurrently. 
The Concurrent-Read, Concurrent-Write (CRCW) PRAM
model allows for both, where, in the priority-write CRCW PRAM version, write
conflicts are resolved by choosing the largest value written, 
and, in the arbitrary CRCW PRAM version, write conflicts are decided arbitrarily.

We measure the performance of PRAM algorithms in terms of \emph{work}, 
which is the total number of operations performed, and \emph{span}, 
which is the maximum number of operations performed by a single processor if we 
execute the algorithm without bounding the number of processors.
This is motivated
by Brent's theorem~\cite{brent1974parallel}, which states that an algorithm with input size $n$, work $W(n)$, and span $T(n)$ can be run on $p$ processors in time $O(W(n)/p + T(n))$. 

\subsection{Our Contributions}
In this paper, we introduce a family of \emph{Synchronous Strict PIP} models. 

\begin{definition}[Synchronous Strict PIP Model]
Let $1 \leq p\leq n$ be the number of processors available. 
A PRAM algorithm is Synchronous Strict PIP if it 
(1) allocates no shared memory aside from the input, 
(2) uses $\Theta(1)$ private memory per processor, 
(3) has $O((n/p) \mathrm{polylog}(n))$ span, 
and 
(4) uses at most $\Theta(1)$ additional space, i.e., is {strictly in-place}, when run sequentially. 
\end{definition}

Observe that the Synchronous Strict PIP model admits synchronous parallel 
algorithms with strictly stronger memory guarantees than those under the 
asynchronous Strong PIP and Relaxed PIP models.
Indeed, Synchronous Strict PIP 
retains the flexibility of the Relaxed PIP model 
by allowing a range of processor counts
while providing space guarantees that are even stronger than the Strong PIP model. 
To the best of our knowledge, there have not been previous algorithms defined in terms of this Synchronous Strict PIP model. 
The work that comes closest is a parallel random permutation algorithm of Hutton and Melrod~\cite{hutton2025encoding}.
When emulated on a CRCW PRAM, 
their algorithm satisfies our conditions (1), (2), and (4), 
but not (3). 

To show the effectiveness of our hardware-based
models, we give several efficient Synchronous Strict PIP algorithms for fundamental parallel problems.
Our results are summarized in \Cref{tab:summary}. 
In Sections~\ref{sec:IJ} and~\ref{sec:more-algos}, 
we introduce the \emph{IJ swap}, an algorithmic technique used throughout the work and give several parallel in-place algorithms analyzed in the work-span framework.
In \Cref{sec:sublinear}, we present our second main contribution, the \emph{parallel-augmented sweep} paradigm. 
This is an algorithmic insight that allows us to efficiently schedule our in-place algorithms analyzed 
in the work-span framework to run on $p < n$ processors while maintaining the same in-place guarantees. 

\begin{table}[t!]
  \begin{center}
  \begin{tabular}{|c|c|c|c|}
    \hline
    {\bf Algorithm} & {\bf Work-Efficient} & {\bf PIP} \\
    \hline
    Prefix$^*$ & \checkmark & Synch. Strict \\
    \hline
    Packing/Partition$^*$ & \checkmark & Synch. Strict\\
    \hline
    QuickSort & \checkmark & Strong\\
    \hline
    List Contraction$^*$ & \checkmark & Synch. Strict\\
    \hline
    Tree Contraction$^*$ & \checkmark & Synch. Strict\\
    \hline
    Rand. Permutation$^*$ & \checkmark & Synch. Strict\\
    \hline
    Convex Hull & \checkmark & Strong\\
    \hline
    Convex Hull Presorted$^*$ &  & Synch. Strict\\
    \hline
    SampleSort$^*$ & \checkmark & Relaxed\\
    \hline
  \end{tabular}
\end{center}
  \caption{A summary of our new parallel in-place algorithms.
  We indicate our main results with $^*$.
  }
  \label{tab:summary}
\end{table}

We also provide an orthogonal result that is of independent interest. 
We apply a novel balls in bins analysis to develop our Relaxed PIP SampleSort, 
the first PIP comparison sorting algorithm with optimal span and optimal work 
(cf.~\cite{axtmann2022engineering,zheng1999efficient}). 
\section{The Indiana Jones Swap}\label{sec:IJ}

In this section, we introduce the \emph{Indiana Jones} (IJ)
swap, which is a technique used throughout this paper. This technique was
inspired by an iconic scene in the movie,
\textit{Raiders of the Lost Ark}, 
in which Indiana Jones attempts to steal a Golden Idol
by swapping it with a bag of sand.
See Figure~\ref{fig:ij}.

\begin{figure}[hbt]
\centering
\includegraphics[width=\columnwidth]{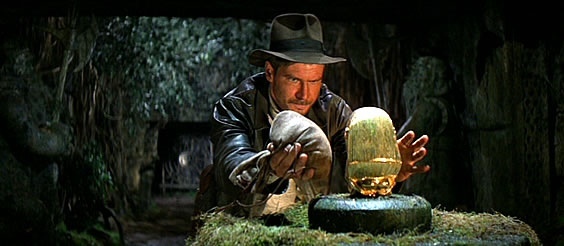}
\caption{
A still from \textit{Raiders of the Lost Ark}. 
\label{fig:ij} 
Directed by Steven Spielberg, Lucasfilm Ltd.~and Paramount Pictures, 1981.
This image is from copyrighted material, the use of which has not 
been specifically authorized by the copyright owner. 
The author(s) have determined this to be fair use of the copyrighted material 
as referenced and provided for in \S~107 of the U.S.~Copyright Law.}
\end{figure}

Our technique similarly involves
each processor strategically swapping data between shared memory
and its private stash. If applied with care, we do not need to
allocate any more shared memory than what is required to store the
input.  For the remainder of this section, we demonstrate the IJ
swap and sub-techniques we call \emph{stashing} and \emph{windowing}
by implementing fundamental parallel algorithms in-place.

\subsection{Parallel Prefix}\label{sec:prefix}

The stashing technique simply involves processors making a copy of the element at their respective indices of the input array, copying it into their 
private stash memory. Such a step requires constant span and work proportional to the number of processors. We frequently use this operation to preserve values that would be destroyed or to free up space for some in-place computation.

Now we consider the parallel prefix problem. 
We begin by observing that the recent Strong PIP prefix algorithm of Gu, Obeya, and Shun~\cite{gu2021parallel} can be easily adapted to our framework. 
Their algorithm is a simple variant of the standard work-efficient tree-based parallel prefix algorithm~\cite{blelloch1990prefix}.
This algorithm performs two passes over the data, the up-sweep and the down-sweep. 
In the up-sweep, we treat each input element as a leaf of a balanced binary search tree. 
The value of each leaf's parent contains the sum of their children, so the root of the tree contains the total sum. 
Each level of the tree can be computed in parallel in $O(1)$ span and work proportional to the number of sums performed. 
Thus, the up-sweep takes $O(\log n)$ span and $O(n)$ work since the number of sums is halved after each level.

In the down-sweep, we use the intermediate sums to compute the true prefix values of the array. Now we traverse from the root down to the leaves. 
We set the root's value to 0. 
In the first round, we pass the sum of the root's value with its left child's value to its right child. We pass the root's original value to its left child. 
This is repeated in the next round for the root's children, and so on for $O(\log n)$ rounds. The final values of each leaf represent the final prefix values. The down-sweep also takes $O(\log n)$ span and $O(n)$ work.

In general, representing this binary tree takes $O(n)$ extra space. Gu, Obeya, and Shun~\cite{gu2021parallel} make the simple observation that we can overwrite values of the original array to maintain the intermediate sums. In the up-sweep, whenever we are summing elements $a$ and $b$, we simply overwrite $b$ with $a+b$. The root of the tree is held at the end of the array. In the down-sweep, the nodes of the ``binary tree'' can be recovered by index manipulation. At depth $i$, each node's right child is at its own index and its left child is $n/2^{i+1}$ indices to the left. They show that both passes still take $O(\log n)$ span and $O(n)$ work. 



We first observe that their algorithm is ``destructive'', in the sense that 
it overwrites the array's original contents after performing the operation. 
We can remedy this by performing an 
initial stashing operation taking $O(1)$ span and $O(n)$ work to preserve 
the original values in each processor's private memory. 
Notably, this allows us to implement packing optimally, which we show in the following section.

Using standard techniques, we observe that their algorithm can be implemented in 
the EREW PRAM model with $O(\log n)$ span and $O(n)$ work. 
Furthermore, we note that computation is only done prior to each recursive call in both the up-sweep and down-sweep of the algorithm. 
This means that processors can simply terminate when they finish work 
at the bottom of the recursive tree 
and so do not have to maintain any information about prior recursive calls.  
Thus, 
if there are $\Theta(n)$ processors, 
each processor only needs $O(1)$ private 
memory to maintain information relevant to their current recursive call as well as the stashed original data.

\begin{theorem}
There exists an algorithm that performs parallel prefix with $O(\log n)$ span 
and $O(n)$ work in the EREW PRAM model when each of $n$ processors 
has $O(1)$ private memory.
\end{theorem}

As Gu, Obeya, and Shun note, the naive serialization of this algorithm gives a sequential algorithm that consumes $O(\log n)$ space due to recursive overhead~\cite{gu2021parallel}. 
In \Cref{sec:sublinear}, we apply the 
parallel-augmented sweep technique to 
generalize the above algorithm such that 
it is efficient and in-place for all processor counts $p \leq n$. 
This naturally implies a strictly in-place sequential algorithm. 


\subsection{Packing and Parallel Partition}\label{sec:packing-partition}
Recently, Kuszmaul and Westover~\cite{kuszmaul2020cache} used inversion encoding to implement a parallel partition algorithm in-place in the EREW PRAM
model. However, their implementation takes $O(\log n \log\log n)$ span, assumes $O(\text{polylog } n)$ private space per processor, 
and requires a complicated workaround to handle duplicate elements. 
In this section, we present a simple algorithm using the IJ swap that has optimal span, requires only $O(1)$ space per processor, and handles duplicate elements automatically. 
As an aside, 
our partition algorithm immediately admits the first optimal, $O(\log^2 n)$-span Strong PIP QuickSort algorithm in the EREW PRAM. 

Before discussing the parallel partition problem, let us first
discuss a related problem
called \emph{packing}. In its simplest form, the packing problem
is given an array of length $n$ containing 0's and 1's in any order
and we wish to pack the 1's such that they appear in the beginning of
the array.  To begin, we apply our in-place prefix sum to the array
and stash the result.  For the entries containing 1's, this indicates
their new location in the final array.  Using the stashed original values,
we can repeat this process but instead have the 0's contribute 1
to the prefix sum and the 1's contribute 0. This gives the 0's their
relative rankings.  To produce their actual rankings, we simply add
the total number of 1's to each 0's rank.  


Lastly, each processor writes their value to the index of the array corresponding to its rank. 
Since each processor can stash their element, we do not have to worry about overwriting other elements. 
The algorithm consists of a constant number of stashing and swapping operations, 
each of which take $O(1)$ span and $O(n)$ work and use at most $O(1)$ words of private memory per processor. 
In addition, the algorithm performs a constant number of in-place parallel prefixes, which take $O(\log n)$ span and $O(n)$ work and use $O(1)$ private memory per processor. 
We conclude with the following. 

\begin{theorem}
There exists an algorithm that performs parallel packing with $O(\log n)$ span
and $O(n)$ work in the EREW PRAM model when each 
of $n$ processors has $O(1)$ private memory.
\end{theorem}

In the parallel partition problem, we are given a list of $n$ elements and a pivot $q$. We wish to return a list such that all elements less than $q$ are to its left in the array and all elements greater than or equal to $q$ are to its right. 
Assume that $q$ is initialized as the last element in the array. 
Observe that we can reduce the partitioning problem to parallel packing in which an element is assigned a 1 if it is smaller than the pivot and a 0 otherwise. 
Thus, we immediately have the following.

\begin{corollary}
There exists an algorithm the performs parallel partition in-place 
in $O(n)$ work and $O(\log n)$ span in the EREW PRAM model when 
each of $n$ processors has $O(1)$ private memory.
\end{corollary}


\subsection{Deterministic Reservations and Windowing}

In this section, we describe our second sub-technique, \emph{windowing}. We do so by implementing in-place algorithms that perform list contraction and tree contraction in polylogarithmic span and linear work for the EREW PRAM. 

We adapt algorithms from Shun, Gu, Blelloch, Fineman, and Gibbons~\cite{shun2014sequential},
who developed their algorithms using a \emph{deterministic reservation} framework, a paradigm in which parallel algorithms hew closely to their sequential counterparts~\cite{shun2014phase,shun2014sequential,blelloch2012internally,blelloch2012greedy,gu2021parallel,shun2013reducing}. The basic idea is to retain the logic of the sequential algorithm, but have each processor attempt to do each step in parallel. 
In most algorithms, this will produce conflicts as processors attempt to do operations concurrently that must be done in a certain order. 
To handle such conflicts, 
algorithms in this framework use \emph{reservations}, 
wherein processors write some data related to the operation they intend to do. 
Operations in conflict may attempt to write to the same place. 
A processor whose reservation is preserved after each step is permitted to do their operation, 
ensuring that no conflicting operations are committed at the same time. 

In the context of list contraction and tree contraction, conflicts occur when two processors attempt to contract the same node. In both problems, Shun {\it et al.}~handle this by initializing a reservation boolean array $R$. In each round of the algorithm, if the processor has an operation to commit, they first reserve their location in $R$. Then, if their reservation holds, they subsequently commit the operation by manipulating data in $L$. See~\Cref{alg:list-contraction}.

\begin{algorithm}[hbt]
	\caption{Non-in-place List Contraction algorithm of~\cite{shun2014sequential}.}\label{alg:list-contraction}
	\begin{algorithmic}[1]
		\State{\textbf{Input:} Array $L$ of size $3n$ containing $n$ nodes and $2n$ pointers to each node's predecessor and successor.}
        \State{\textbf{Output:} Contracted list $L$.}
        \State $R = \{0, \ldots, 0\}$\Comment{boolean array}

        \Procedure{reserve}{$i$}
            \If{$i < L[i].\text{prev}$ and $i < L[i].\text{next}$}
                \State{$R[i] = 1$}\Comment{reserve own location}
            \EndIf
        \State \Return{1}
        \EndProcedure
        \Procedure{commit}{$i$}
            \If{$R[i] = 1$}
                \If{$L[i].\text{prev} \neq \text{null}$}
                    \State $L[L[i].\text{prev}].\text{next} \leftarrow L[i].\text{next}$
                \EndIf
                \If{$L[i].\text{next} \neq \text{null}$}
                    \State $L[L[i].\text{next}].\text{prev} \leftarrow L[i].\text{prev}$
                \EndIf
                \State \Return{0}
            \Else
                \State \Return{1}
            \EndIf
        \EndProcedure
	\end{algorithmic}
\end{algorithm}

In Shun {\it et al.}'s list contraction algorithm~\cite{shun2014sequential}, the main operation that is not in-place involves the allocation of $R$. We make the trivial observation that, rather than allocating a new array $R$, we can simply use $L$ to store reservations. Have processor $p_i$ temporarily swap $L[i]$ with its reservation value (either a 1 or 0) after the \textsc{reserve} phase. Then, after the first if statement in the \textsc{commit} phase, we can swap the original values of $L$ back in and proceed as normal. We call this temporary revealing of metadata \emph{windowing}. 

After each round of reserving and committing in~\Cref{alg:list-contraction}, we perform a packing step to compact all the processors that have not yet committed their operation~\cite{shun2014sequential}. 
We can free up space to do this by temporarily stashing the first $p$ elements in $L$, where $p$ is the current number of active processors, and then applying our in-place packing algorithm from \Cref{sec:packing-partition}. 

An analysis by Shun {\it et al.}~\cite{shun2014sequential} shows that this algorithm takes $O(\log n)$ rounds with high probability in $n$,\footnote{
An algorithm succeeds with high probability in $n$ (w.h.p.) if it succeeds with probability greater than $1 - 1/n^c$ for some constant $c \ge 1$.} 
assuming the ordering of $L$ is randomized. 
Standard techniques can be used to ensure that all reads and writes are exclusive.
Thus, we have the following.

\begin{theorem}
There exists an algorithm that solves list contraction in the EREW PRAM model
when each of $n$ initial 
processors has $O(1)$ private memory in $O(\log^2 n)$ span w.h.p. and $O(n)$ work in expectation, assuming the list $L$ is in random order.
\end{theorem}

Their tree contraction algorithm has an almost identical structure~\cite{shun2014sequential}, so the same technique implies the following.

\begin{corollary}
There exists an algorithm that solves tree contraction (on complete binary trees with randomly-ordered leaves) in the EREW PRAM model when each 
of $n$ initial processors has $O(1)$ private memory in $O(\log^2 n)$ span w.h.p. and $O(n)$ work in expectation.
\end{corollary}



\section{Efficient Parallel Algorithms with Constant Space Per Processor}\label{sec:more-algos}

In this section, we demonstrate the IJ swap further by implementing non-trivial algorithms for in-place random permutations, 2D convex hulls, and sorting.

\subsection{Random Permutation}\label{sec:random-perm}
We now adapt another deterministic reservations algorithm of 
Shun, Gu, Blelloch, Fineman, and Gibbons~\cite{shun2014sequential}, 
this time for computing a random permutation of an array, $A$. 
Their algorithm is a parallelization of the sequential Knuth shuffle algorithm for the priority-write CRCW PRAM. 
In the Knuth shuffle, we iterate backwards through the array and 
repeatedly swap current element $A[i]$ with some randomly selected element from indices 0 to $i$. 
Shun {\it et al.}'s parallelized version works like so. Given $n$ processors, 
their algorithm assigns each processor to an array index $i$. 
We have them each generate $H[i]$, a random value between $0$ and $i$. Each processor's goal is to commit a swap between $A[i]$ and $A[H[i]]$. 
However, swaps that overlap the same index are in conflict and cannot be done concurrently. 
To solve this, the algorithm of Shun {\it et al.} uses reservations. 

At the beginning of each round, we initialize an array of reservations, $R$,
in which all values are set to $-1$. Each processor $p_i$ ``reserves'' their swap by writing $i$ to $R[i]$ and $R[H[i]]$. This is a priority write operation, so the maximum value of $i$ written concurrently is preserved. Each processor $p_i$ then reads $R[i]$ and $R[H[i]]$. If their two reservations were preserved, then processor $p_i$ commits their swap in $A$. Shun {\it et al.} show that the depth of these conflicts is at most $O(\log n)$ w.h.p.~\cite{shun2014sequential}. As a result, after $O(\log n)$ rounds, every processor will have committed their swap w.h.p. 
Once again,~\cite{shun2014sequential} use packing after each round to limit the work of inactive processors. 

Using the IJ swap, we can make their algorithm in-place.
Just like our list contraction and tree contraction algorithms, we can perform windowing on the input array $A$ to simulate $R$. 
However, we must be more careful since processors can read and write to multiple, sometimes overlapping, locations of $R$ rather than just a single location.
Each active processor $p_i$ begins each round by stashing $A[i]$ and $A[H[i]]$ in private memory. Then each $p_i$ writes a $-1$ to $A[i]$ and $A[H[i]]$ concurrently. 
Since each index that an active processor intends to reserve is set to a $-1$, we can act as if $A$ is the reservation array $R$.
Next, each $p_i$ performs a priority write to reserve $A[i]$ and $A[H[i]]$ in parallel. 
If a processor observes that its reservation was overwritten, it immediately writes its stashed values back to their original spots. 
Finally, the processors with successful reservations write their stashed values back in swapped order. 

As with list and tree contraction, we use our in-place packing method after each round. 
Stashing and windowing add $O(1)$ span to each round. 
Since only active processors perform such steps, we can charge the extra work done as a constant multiple of the work done in the original algorithm to read and write reservations. As a result, total work increases by no more than a constant factor. 
In combination with a runtime analysis by 
Shun {\it et al.}~\cite{shun2014sequential}, we conclude the following.

\begin{theorem}
There exists an algorithm that computes a random permutation in 
the priority-write CRCW PRAM model when each of $n$ initial
processors has $O(1)$ private memory in $O(\log^2 n)$ span w.h.p. and $O(n)$ work in expectation.
\end{theorem}

\subsection{Convex Hull}\label{sec:convex-hull}
We next show how to construct the upper hull of a set of $n$ points in the plane in-place, assuming the input is sorted by its $x$-coordinate. 
Our algorithm requires $O(\log n)$ span and $O(n\log n)$ work in the 
CREW PRAM model. The lower hull can be computed symmetrically.\footnote{
If points are sorted radially, 
the entire convex hull can be recovered at once (see, e.g.,~\cite{preparata1979optimal}). 
We consider the upper hull for simplicity.
}
Our key insight is to represent each intermediate hull as a permutation of its hull points, rather than in an external data structure. 
Indeed, this technique is used often in sequential in-place algorithms for computational geometry~\cite{bronnimann2004towards,bose2007space,chen2003space,blunck2010place,chan2008place,bronnimann2002place}. 
Surprisingly, we believe we are the first to do so in parallel 
computational geometry. 
We also present a method to maintain recursive metadata despite having only constant space per processor.
For simplicity, we assume that points are in general position and that no two points share the same $x$-coordinate. These assumptions may be lifted via standard techniques~\cite{mehlhorn2006reliable}. 

Our algorithm is based on a classic parallel convex hull 
algorithm~\cite{goodrich1987efficient,aggarwal1988parallel}. 
At a high-level, the algorithm applies parallel divide-and-conquer. The input to the algorithm is a set of $n$ points stored in an array ordered by their $x$ coordinates.
The algorithm first splits the input into $\sqrt n$ equal-sized subproblems. 
We recurse and return with $\sqrt n$ intermediate upper hulls that must be combined into a single upper hull. 
We use the double binary search method of Overmars and van Leeuwen~\cite{overmars1981maintenance} to compute upper tangents across all pairs of hulls and use this information to determine what portions of each intermediate hull should be discarded. A final prefix operation concatenates the surviving portions of each intermediate hull. 

The original algorithm uses extra space to maintain each intermediate hull such that they can be binary-searched, split, and joined in parallel. For example, they could be represented as binary search trees ordered by $x$-coordinate. 
We cannot afford this extra $O(n)$ space. 
Instead, we maintain upper hulls of each subproblem in-place by rearranging points in the input array. 

\begin{figure}
    \centering
    \includegraphics[width=0.9\linewidth]{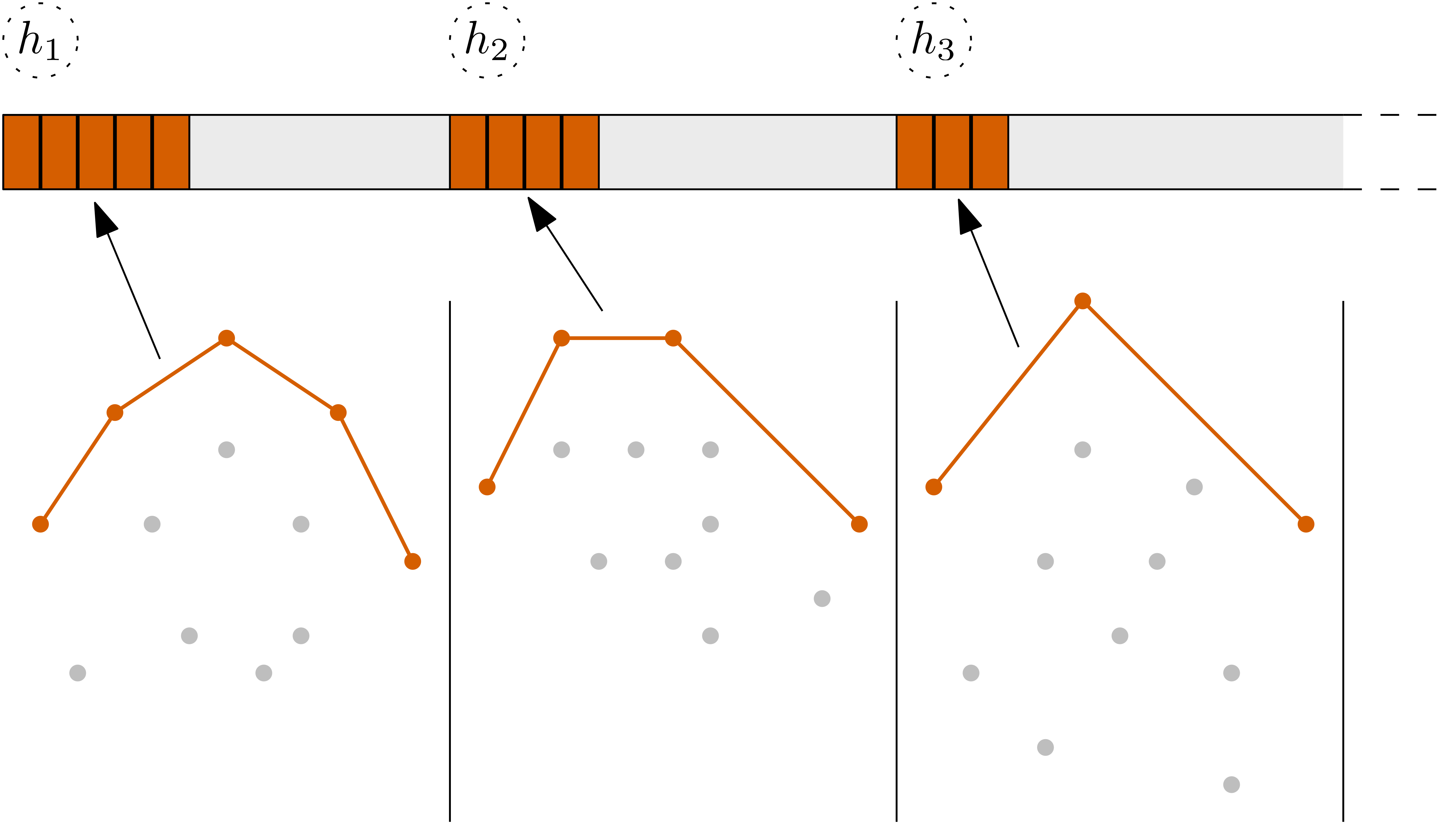}
    \vspace*{-6pt}
    \caption{The combine step for our in-place convex hull algorithm. 
Intermediate hulls are represented in-place with hull points packed the left of their subarray. 
}
    \label{fig:hull}
\end{figure}

For example, say that subproblem $S_i$ has an upper hull of $h_i$ points. Let us assume that, when $S_i$ returns to its parent subproblem, the first $h_i$ entries in $S_i$'s section of the array represent its upper-hull points ordered by their $x$-coordinate (the remaining $|S_i| - h_i$ points may be in arbitrary order). 
Furthermore, let us assume that the first processor assigned to $S_i$ has $h_i$ stored in its private stash. See \Cref{fig:hull}. Assuming this is true for all subproblems, we will show that we can construct the upper hull of their parent subproblem in-place in $O(\log m)$ span and $O(m\log m)$ work, where $m = \sum_i |S_i|$. 

The first step in combining all subproblems is to assign each processor a pair of upper hulls for which to compute an upper tangent. This pairing can be computed independently by each processor by examining their index with respect to the current problem size. 
Then each processor must perform a double binary search on a pair of hulls. The first processors of each subproblem $S_i$ can use windowing to temporarily reveal $h_i$ in the first $\sqrt m$ array indices of the parent subproblem. Processors that access $S_i$'s upper hull can stash $h_i$ and use it to determine the middle element of $S_i$'s hull. Then the first processor of each subproblem can swap $h_i$ back into their private memory and we can perform concurrent double binary searches as usual on each hull's packed array representation.

After each processor determines an upper tangent, they can stash their 
index's point and write the slope of the computed tangent to their index. Processors from each subproblem $S_i$ can in parallel perform two max-finds to determine the largest-slope 
tangent lines coming from the left 
and from the right of $S_i$ via our in-place prefix algorithm. 
With this information, each intermediate hull is reduced to a (possibly empty) convex chain, 
which we must combine into a single upper hull of $h$ points which represents 
the combined subproblems. 
This can be done easily via in-place packing. 
Each processor assigns a 0 to their point if it is not part of the new upper hull and a 1 if it is. 
By properties of packing, our packed array preserves the ordering of the hull points. 
The first processor stashes the number of hull points, $h$, and we return. 

\subsubsection{Recursive metadata load-sharing}

The last issue is how to maintain metadata about recursion such that processors can efficiently return to parent problems. Naively, we can have each processor maintain information about each subproblem they are a part of, but this blows up space usage to $\alpha\log\log n$ per processor for some constant $\alpha$, where $\alpha\log\log n$ is the number of levels of recursion. Instead, we can spread the load across several processors. 

\begin{figure}
    \centering
    \includegraphics[width=\linewidth]{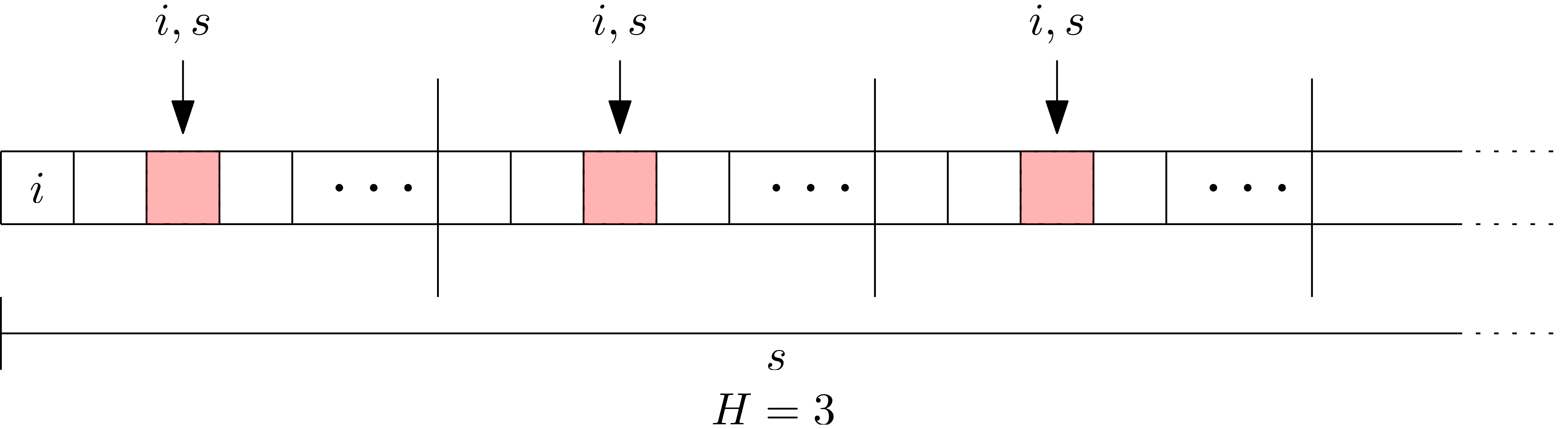}
    \caption{An example of recursive metadata load sharing. 
    We are at a parent problem at recursive depth $H=3$ about to recurse into several child subproblems. 
    Metadata (starting index and problem size) for the 
    parent subproblem is stored at index $3$ of each child subproblem. 
    By definition of the recursion and by the fact that all subproblems are size $> \alpha\log\log n$, 
    each index of our input array, and so each processor, 
    is associated with
    exactly one metadata storage.
    }
    \label{fig:recursive-metadata}
\end{figure}

First, we note that all subproblems of size $\leq \alpha\log\log n$ can be computed via in-place Graham scan in $O(\log\log n)$ span and $O(n)$ 
work since the points are already sorted~\cite{bronnimann2002place}. Thus, doing so does not increase the asymptotic complexity of the algorithm. Let $H$ be our current level of recursion, starting with $H=1$. Say that we are at some parent subproblem of size $m$ and depth $H$ and about to recurse into its children. Just before recursing, have the $H$-th processor of each child subproblem store two values: the position of the first processor of the parent problem and $m$. 
See Figure~\ref{fig:recursive-metadata}.

After computing the upper hulls of each child subproblem and just before returning to the parent problem, 
have the $H$-th processor in each child subproblem window the position and size information of the parent problem. 
Each processor can maintain their current depth $(H+1)$ with constant storage, so it takes $O(1)$ span and $O(m)$ total work across all the child subproblems to 
broadcast the parent problem's position and size information.



From the observation above, we can assume each subproblem above the base case is size $> \alpha\log\log n$. 
Then by construction of the original algorithm and our load-sharing procedure, each index of the array is associated with at most a single metadata storage.
As a result, each processor stores at most a constant amount of recursive metadata.

The cost of the rest of the combine step in the parent subproblem is dominated by the  binary searches, which take $O(\log m)$ span and $O(m\log m)$ work.
These time bounds match the original algorithm's  bounds asymptotically. We use their runtime analysis to conclude that the  span of the algorithm is $O(\log n)$ and its  work is $O(n\log n)$~\cite{goodrich1987efficient,aggarwal1988parallel}.

\begin{theorem}
There exists an algorithm to compute the 2D convex hull of $n$ presorted points in the plane in $O(\log n)$ span and $O(n\log n)$ work for the CREW PRAM when each of $n$ processors has $O(1)$ private memory. 
\end{theorem}

\subsection{SampleSort}

In this section, we apply an analysis of a new variant of parallel balls in bins to prove the following.

\begin{theorem}\label{thm:sample-sort}
We can implement an algorithm that sorts $n$ elements in $O(\log n)$ span with high probability
and $O(n\log n)$ expected work for the arbitrary CRCW PRAM when each of $n$ processors has $O(1)$ private memory.
\end{theorem}

SampleSort is a generalization of randomized quicksort in which we sample more than one pivot~\cite{frazer1970samplesort,blelloch1998experimental,axtmann2022engineering}. We can sort the pivots and then partition each remaining element into ``buckets'' based on where they lie in between the pivots. We then recurse into each of these buckets and repeat. Our in-place algorithm is based on the non-in-place sample sort algorithm of Blelloch, Fineman, Gu, and Sun~\cite{blelloch2020optimal} 
(see also, Exercise 9.30 of~\cite{jaja1992parallel}). 
At each recursive step of size $n$, their algorithm randomly samples $n^{1/3}\log n$ elements, sorts them, and sub-samples $O(n^{1/3})$ of them to define the buckets. 

Sorting the sampled elements can be done in-place using the IJ swap in $O(\log n)$ span and $o(n)$ work since storing each pairwise comparison takes $O((n^{1/3}\log n)^2) \subset o(n)$ space. Thus, we can free up space in the input array by stashing and use this space to maintain and process all pairwise comparisons. 
In addition, performing parallel binary search to determine which elements belong to which buckets can be done in $O(\log n)$ span and $O(n\log n)$ work by packing the $O(n^{1/3})$ pivots at the beginning of the array in sorted order. 

Blelloch {\it et al.}~\cite{blelloch2020optimal} prove the following helpful lemma, bounding the maximum size of each bucket. 

\begin{lemma}\label{lem:high-prob-bucket}
There exists a constant $c$ such that the number of elements falling into any bucket is no more than $cn^{2/3}$ with high probability.
\end{lemma}
\begin{proof}
See Lemma 4.1 of~\cite{blelloch2020optimal}.
\end{proof}

The authors allocate arrays of size $cn^{2/3}$ to serve as the buckets between each pivot. They can then use a randomized algorithm to write each bucket's elements to the correct place. Since the space allocated to each bucket is at least as large as the number of elements in that bucket, it is simple to show that this is efficient. However, this requires $O(n)$ additional words of shared memory. As a result, we apply a different strategy. 

\subsubsection{Distributing Elements In-Place}

\begin{figure}
    \centering
    \includegraphics[width=.9\linewidth]{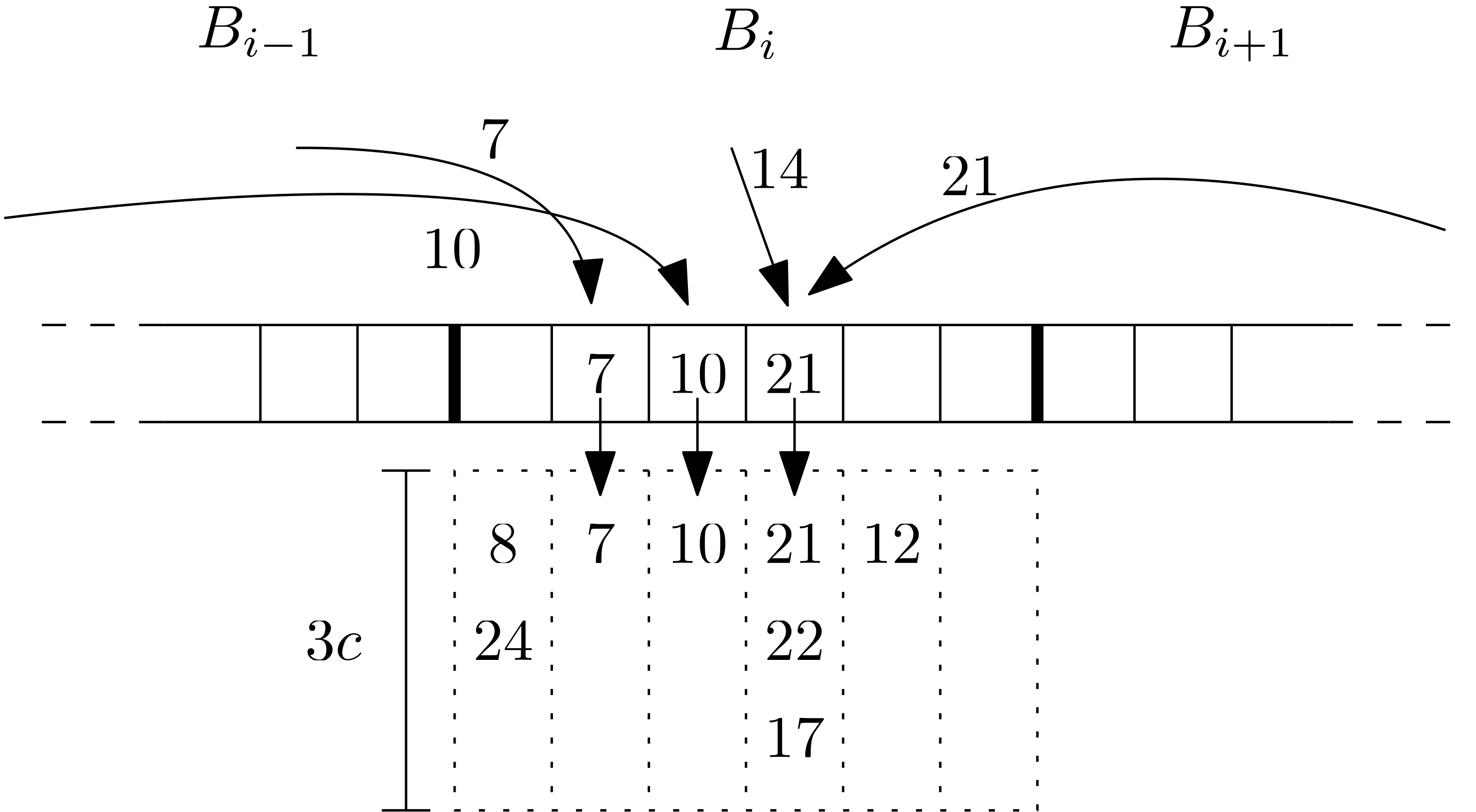}
    \caption{Our in-place distribution algorithm for sample sort. 
    At each step, every element that belongs in $B_i$ tries to write itself to a random location in $B_i$. 
    Write conflicts (e.g., with 14 and 21) are resolved arbitrarily. 
    At the end of the round, each entry copies the element that lands there into its stash. It does so until it reaches its stash size of $3c$. 
    }
    \label{fig:sample-sort}
\end{figure}

Our in-place algorithm to distribute elements into buckets begins by having each processor stash the element originally at their index. We then partition the input array evenly into groups of size roughly $n^{2/3}$, each of which represents a bucket. However, as we noted above, each bucket may contain up to $cn^{2/3}$ elements. Thus the array itself may not have enough space to contain the bucket's elements. We can solve this by using each processor's $O(1)$-sized stash to augment the capacity of each bucket. 

Our method proceeds in rounds. See \Cref{fig:sample-sort} for a visual of the following description. Each processor takes their initially stashed element $e$ and attempts to write it to a random location $j$ in the section of the array corresponding to the bucket it belongs to. Write conflicts are dealt with arbitrarily. 
After the write step, each processor $p_i$ reads the element at their own index $i$. They copy $A[i]$ into their stash so long as they have less than $3c$ elements already stashed. If they are at capacity, they write a null value to $A[i]$ to indicate that they did not stash anything.
Processors then read $j$ to determine if they have successfully written their element $e$. 

By augmenting the capacity of each bucket index with $3c$ elements of private storage, each bucket can store $3cn^{2/3}$ elements. By \Cref{lem:high-prob-bucket}, this is enough space to store all elements w.h.p. in $n$. 

After the distribution process completes, each processor counts the number of elements they have stashed. Similar to parallel partition, we can perform a prefix sum on these counts. The information can be used to assign each element a unique position in the array such that all elements in the same bucket are contiguous. Each element writes to those locations in the input array and then we recurse into each bucket. These final steps take $O(\log n)$ span and $O(n)$ work.

Each round of the distribution process takes $O(1)$ span and $O(n)$ work. 
It remains to show that our distribution algorithm terminates in $O(\log n)$ rounds w.h.p.

\subsubsection{Congested Parallel Balls into Bins}

Our algorithm can be mapped to a parallel variant of balls into bins, where the balls are the elements that need to be distributed and the bins are the array indices. Several works have already studied variants of parallel balls into bins, often in the context of load balancing tasks on machines. See, e.g.,~\cite{stemann1996parallel,berenbrink2016self,adler1995parallel,lenzen2011tight,berenbrink1997allocating,berenbrink2012multiple,adler1998analyzing,lenzen2019parallel}. 
However, it appears that our setting is unique in that all remaining balls are thrown at once yet bins accept at most one ball each round, 
even if their remaining capacity is more than one. 

We call this variant the \emph{congested parallel balls into bins} problem. 
Formally, we have $m$ balls and $b$ bins.
Each bin has a threshold $t \geq m/b$ and a \emph{congestion factor} $g$. 
Each round, we throw all remaining balls uniformly into the $n$ bins. 
Bins try to accept as many balls that land in it, 
subject to the following constraints. 
Bins can accept no more than $g$ balls each round, 
and bins can accept no more than $t$ balls in total. 
Balls that are not accepted by a bin must be thrown again in the next round. 
For a given instantiation of $m$, $b$, $t$, and $g$, 
we are interested in knowing how many rounds it takes until all balls are accepted by a bin w.h.p. in the number of bins $b$.

Let us first focus on a single bucket. The matching instance of the congested parallel balls into bins problem has $b=n^{2/3}$ be the number of bins available, 
$m \leq cb$, $t = 3c$, and $g=1$. We will show that after $O(\log b)$ rounds, all $m$ balls will be allocated w.h.p. in $b$. 

Our approach defines two (implicit) phases of the distribution process. 
The first phase concerns the first $3c$ rounds of the process. 
Since each bin has capacity $3c$, every 
bin in the first $3c$ rounds must accept a ball if it receives one. 
Due to this property, 
we show that a large majority of the balls are placed in this phase. 
The remaining rounds form the second phase, 
in which all straggler balls are placed. 
We show 
that the second phase takes an additional 
$O(\log b)$ rounds w.h.p. in $b$. 

\begin{lemma}\label{lem:first-rounds}
After $3c$ iterations, there will be at most $b/2$ balls remaining w.h.p. in $b$, the number of bins. 
\end{lemma}

\begin{proof}
We first observe that, since $g=1$ and $t=3c$, no bin can exceed its threshold for the first $3c$ rounds. As a result, any bin that receives one or more balls in the first $3c$ rounds must accept a ball. 

Let  $m' \leq m$ be the number of balls thrown at the beginning of a given round. Each ball lands in a bin with probability $1/b$. 
The chance a bin is not hit by any ball is $(1 - 1/b)^{m'}$. 
Assume that $m' > b/2$ (otherwise, we are done). Then $(1 - 1/b)^{m'} \leq (1 - 1/b)^{b/2} \leq 1/\sqrt e$. 
As a result, we expect at most $b/\sqrt e$ bins to not be hit by a single ball each round. 

Let $X$ be a random variable representing the number of bins not hit this iteration.
Applying a Chernoff bound, we see that 
$$\Pr[X \geq (1+\delta)\times b/\sqrt e] \leq e^{-b\delta^2/4\sqrt e}.$$
This holds with exponential probability\footnote{
An event succeeds with exponential probability in $n$ if it succeeds with probability $1 - c_1^{-c_2n}$, for some constants $c_1 > 1$, $c_2$.
} 
in $b$ for any $\delta > 0$. For concreteness, we set $\delta = 2\sqrt e/3 - 1$. This implies that no more than $2/3$ of the bins will fail to be hit with exponential probability in $b$. 

As a consequence, when $m' > b/2$, at least $b/3$ bins will be hit with one or more balls with high probability. Since no bin hits their threshold until round $3c$ at the earliest, this means that bins accept at least $b/3$ balls in a given round w.h.p. 
By union bound over the first $3c$ rounds, 
this holds with probability $1 - 3ce^{-O(b)}$ over all rounds, which is still exponential in $b$. $m - 3c\times b/3 \leq cb-3c\times b/3 < b/2$, so we have the desired result.
\end{proof}

\begin{lemma}\label{lem:remaining-balls}
When there are $b/2$ balls remaining, it takes $O(\log b)$ rounds until all balls are accepted by a bin w.h.p. in $b$.
\end{lemma}
\begin{proof}
We will first show that each ball has a constant probability of being accepted by a bin each round. 
Because our threshold is set to $3c$ and $m\leq bc$, no more than $b/3$ bins can ever be at their threshold. 
For simplicity, we assume that exactly $b/3$ bins are at their threshold and cannot accept any new balls. 

Since the balls are thrown uniformly, each lands in a bin with space available with probability at least $2/3$. Let us assume that a ball must land by itself in an available bin to be accepted.\footnote
{This is a pessimistic assumption. 
In reality, the ball might get lucky and get picked arbitrarily if it lands with one or more other balls.} 
If a pair of balls is thrown uniformly at $b$ bins, they collide with probability $1/b$. 
If we fix one of the balls, the probability that it collides with another ball is at most $(b/2 - 1)/b < 1/2$. Thus, each ball lands by itself in an available bin with probability at least $1/3$.

We can model this event with biased coin flips. Assume we flip a coin with success probability $1/3$ $O(\log b)$ times. The probability that no heads occur is $(2/3)^{O(\log b)} < 1/b^{-\beta}$ for some constant $\beta > 1$ dependent on the constant in the $O$-notation. Then a given ball will be accepted in $O(\log b)$ rounds with probability $1 - b^{-\beta}$. 

By union bound over all $b/2$ balls, the probability of failure is increased to at most $b^{-(\beta - 1)} / 2$. 
Adjusting the constants in the $O$-notation, we can ensure that $\beta > 2$ and retain the high-probability bound.
\end{proof}

\begin{theorem}\label{thm:sample-sort-rounds}
Our algorithm to distribute elements in-place succeeds after $O(\log n)$ rounds w.h.p. in $n$.
\end{theorem}
\begin{proof}

Follows from \Cref{lem:first-rounds} and \Cref{lem:remaining-balls}. Since $b=n^{2/3}$, small adjustments in constant factors of the number of rounds can be made to ensure that both lemmas succeed w.h.p. in $n$, even after taking a union bound over all $O(n^{1/3})$ buckets.





\end{proof}

By \Cref{thm:sample-sort-rounds} and our initial discussion, our algorithm takes $O(\log n)$ span w.h.p. and $O(n\log n)$ expected work per recursive step of size $n$. By the runtime analysis of~\cite{blelloch2020optimal}, our algorithm satisfies \Cref{thm:sample-sort}. 
\section{The Parallel-Augmented Sweep Paradigm: Synchronous Strict PIP Scheduling}\label{sec:sublinear}

In this section, we introduce the \emph{parallel-augmented sweep} paradigm, a new technique that combines 
our in-place parallel algorithms defined above with strictly in-place sequential sweep algorithms
to produce efficient, 
in-place parallel algorithms for all processor counts $p \leq n$.
Our technique is a generalization of the \emph{decomposability} property described by Gu, Obeya, and Shun~\cite{gu2021parallel}. 

A problem is decomposable if, given $p$ processors, 
we can sweep the input in blocks of size $p$ and compute each block independently to form a correct solution. 
Our technique goes further and handles problems in which these solved subproblems are not independent and must be combined to produce a globally correct solution.

In addition, the algorithms of~\cite{gu2021parallel} that leveraged this property were Relaxed PIP, as they all required allocating $O(p)$ extra words of shared memory. 
In contrast, the algorithms we present in this section allocate no extra shared memory, consume $O(1)$ private memory per processor, and elegantly degrade to standard, strictly in-place sequential algorithms when $p$ is set to 1. 

\subsection{Parallel Prefix}

As a warm-up, we adapt a well-known algorithm for prefix sums when the number of processors $p < n$ to produce a Synchronous Strict PIP prefix algorithm. 
The algorithm works like so. 
Split $n$ into $p$ groups of size $n/p$, assign each processor to a group, and compute each group's sum iteratively.
We then compute a prefix sum on these intermediate sums in parallel using the algorithm from \Cref{sec:prefix}. 
Each processor can then return to its group and infer the prefix values of the remaining $n/p - 1$ elements using the results of the parallel prefix operation. 


\begin{theorem}
There exists an algorithm that performs parallel prefix in $O(n)$ work and $O(n/p + \log p)$ span given $O(p)$ processors for the EREW PRAM that each have $O(1)$ private space.
\end{theorem}

When $p$ is 1, this algorithm becomes the trivial sweep algorithm for prefix sums. 

\begin{corollary}
There exists a Synchronous Strict PIP parallel prefix algorithm.
\end{corollary}

\subsection{Packing}\label{sec:packing-sublinear}

Recall that in the packing problem we are given an array of length $n$ containing 1's and 0's. The goal is to move all the 1's to the beginning of the array. 

We divide the input into blocks of size $p$ and 
perform parallel packing on the first block using our algorithm from \Cref{sec:packing-partition}. 
Let $i$ be the index of the first 0 in this packed block. We use a prefix copy operation to share $i$ with all processors. 

In the next step, we perform packing on the second block. Let $j$ be the index of the first 0 after packing the second block. We then merge the two blocks such that the first $2p$ elements of the array are packed. 
To merge the two blocks, we simply move all the 1's from the second block and place them into the first block, starting at $i$. 
Each block is size $p$, so we have enough processors to copy all 1's in the second block at once. 
Any 0's that were overwritten are placed in the original location of the shifted 1's.

This merge step can be done in constant span and $O(p)$ work with exclusive reads and writes since each processor can compute the at most three indices they must read from and write to independently from $i$, $j$, the starting index of the current block, and their processor number. After this merge step, we update $i$ to be the location of the first 0 among the merged blocks and then repeat for the next block. 

There are $n/p$ blocks, and processing and merging each block takes $O(\log p)$ span and $O(p)$ work, dominated by the cost of performing parallel packing on that block. 
Thus, the algorithm takes $O(\frac n p \log p)$ span and $O(n)$ work in total. 

\begin{theorem}
There exists an algorithm that performs parallel packing in $O(n)$ work and $O(\frac n p \log p)$ span given $O(p)$ processors for the EREW PRAM that each have $O(1)$ private space.
\end{theorem}

When $p$ is 1, this algorithm becomes the following standard sequential packing algorithm. We maintain two pointers, $i$ and $j$, which both begin at the first 0 in the array. At each iteration, advance $j$ forward one index. If $j$'s index now holds a 1, swap its value with $i$'s and advance $i$ forward by one. Otherwise, do nothing. 

By the same observation we make in \Cref{sec:packing-partition}, we also have the following.

\begin{corollary}
There exists an algorithm that performs parallel partition in $O(n)$ work and $O(\frac n p \log p)$ span given $O(p)$ processors for the EREW PRAM that each have $O(1)$ private space.
\end{corollary}

\begin{corollary}
There exist Synchronous Strict PIP parallel packing and parallel partition algorithms.
\end{corollary}

Thus, our algorithm provides a strong memory-usage guarantee for all $p$
and shaves a loglog factor in span over the previous best in-place parallel partition algorithm by Kuszmaul and Westover~\cite{kuszmaul2020cache}.


\subsection{Deterministic Reservations}

Gu, Obeya, and Shun~\cite{gu2021parallel} made the simple observation that random permutation, list contraction, and tree contraction are decomposable.
As described above, algorithms with this property are a special case of 
our parallel-augmented sweep paradigm as the subproblems do not need to be merged after being processed. 
As a result, we immediately have the following.

\begin{theorem}
There exist algorithms that perform list and tree contraction in $O(n)$ expected work and $O(\min\{\frac n p\log p\log n, n\log p\})$ span w.h.p. given $O(p)$ processors for the EREW PRAM that each have $O(1)$ private space. Futhermore, there exists an algorithm that computes a random permutation of $n$ values in $O(n)$ expected work and $O(\min\{\frac n p\log p\log n, n\log p\})$ span w.h.p. given $O(p)$ processors for the priority-write CRCW PRAM that each have $O(1)$ private space.
\end{theorem}
\begin{proof}
Follows from a sweep of the input array, applying our in-place algorithms for the respective problems to successive blocks of size $p$.
Recall that the conflict depth of these algorithms is $O(\log n)$ w.h.p. if processing all $n$ entries at once~\cite{shun2014sequential}. Then the conflict depth of each $p$-sized subproblem can be no more than $O(\log n)$ w.h.p. since each processes a subset of all conflicting nodes. 

We note that conflict depth is also surely at most $O(p)$, giving us an alternate span bound of $O(\frac n p \log p \times p) = O(n\log p)$ that holds deterministically. In particular, this implies that the algorithm has $O(n)$ span when $p = O(1)$.
\end{proof}

When $p$ is 1, these problems reduce to their respective standard in-place sweep algorithms 
(see, e.g.,~\cite{shun2014sequential}). 

\begin{corollary}
There exist Synchronous Strict PIP algorithms for parallel list contraction, tree contraction, and random permutation
\end{corollary}


\subsection{2D Convex Hull}

In this section, we adapt our packing algorithm to show the following.

\begin{theorem}\label{thm:convex-hull-proc}
There exists an algorithm to compute the convex hull of $n$ presorted points in the plane in $O(\frac n p \log n)$ span and $O(n\log p + \frac n p \log n)$ work given $O(p)$ processors for the CREW PRAM that each have $O(1)$ private space. 
\end{theorem}

Assume the points have been sorted by $x$-coordinate.
We will compute the upper hull 
(again, we note that the entire hull can be computed if the points are sorted radially).
We once again sweep the array from left to right in blocks of size $p$.
We begin by computing the hulls of the first two blocks separately using our in-place convex hull algorithm of \Cref{sec:convex-hull}. Recall that the algorithm permutes the input such that all the hull points are packed to the left of their section of the array in sorted order. 
To merge these hulls, we then have a single processor compute the tangent points between them using the method of~\cite{overmars1981maintenance}.
Say that the tangent points for the first and second hull are stored in indices $i$ and $j$ respectively.

Now we must remove the hull points between the tangent points and then concatenate the remaining points. By convexity and by definition of our algorithm, the surviving points on each hull are contiguous in the array. In addition, since the leftmost point must be on the convex hull, the surviving portion of the first hull is still packed against the left end of the array. 

It remains to move the surviving portion of the second hull so that the merged hull is stored contiguously. 
Since the number of points on the surviving portion of the second hull is no more than $p$, this final step reduces to the merge procedure of our packing algorithm. 

We can compute the convex hull of each block in-place in $O(\log p)$ span and $O(p\log p)$ work. Computing tangents is done sequentially in $O(\log n)$ time. Moving the points in the right hull takes $O(1)$ span and $O(p)$ work. Thus, over all $n/p$ blocks, our algorithm takes $O(\frac n p\log n)$ span and $O(n\log p + \frac n p\log n)$ work and so satisfies~\Cref{thm:convex-hull-proc}.

When $p$ is 1, this algorithm is an in-place variant of Preparata's online convex hull algorithm~\cite{preparata1979optimal}. 

\begin{corollary}
There exists a Synchronous Strict PIP algorithm to construct a 2D convex hull with presorted points. 
\end{corollary}

\begin{corollary}
There exists a Strong PIP algorithm to construct a 2D convex hull of $n$ unsorted points.
\end{corollary}
\begin{proof}
Follows by combining the above algorithm with a Strong PIP sorting algorithm~\cite{gu2021parallel,kuszmaul2020cache,hutton2025encoding}.
\end{proof}


\renewcommand{\emph}[1]{\textit{#1}}
\bibliographystyle{IEEETran}
\bibliography{IEEEabrv, IEEEexample, refs}

\end{document}